\documentclass[letterpaper, 10 pt, conference]{ieeeconf}  

\IEEEoverridecommandlockouts                              

\usepackage{amsmath, amssymb}
\usepackage{xcolor}
\usepackage{graphicx}
\usepackage{mathtools}
\usepackage{dsfont}

\usepackage{algorithm}

\newcommand{\R}{\ensuremath{\mathbb R}}  
\newcommand{\N}{\ensuremath{\mathbb N}}  
\renewcommand{\k}{\mathsf{k}}

\newcommand{\ceil}[1]{\lceil #1 \rceil}

\newtheorem{definition}{Definition}
\newtheorem{theorem}{Theorem}

\newtheorem{assumption}{Assumption}

\newcommand{\rx}{r_{\mathcal{X}}}

\title{\LARGE \bf
Stability of data-driven Koopman MPC with terminal conditions*
}

\author{Irene Schimperna$^{1}$, Lea Bold$^{2}$, Johannes Köhler$^{3}$, Karl Worthmann$^{2}$ and Lalo Magni$^{1}$
\thanks{*I.S.\ gratefully acknowledges funding by the Project “Adaptive Personalised Safe Artificial Pancreas for children and adolescents (APS-AP)” - CUP F53D23000720006 - Grant Assignment Decree No. 960 adopted on 30/06/2023 by the Italian Ministry of University and Research (MUR).
L.B.\ and K.W.\ gratefully acknowledge funding by the German Research Foundation (DFG; project numbers~$535860958$ and~$545246093$).}
\thanks{$^{1}$Irene Schimperna and Lalo Magni are with the Department of Electrical, Computer and Biomedical Engineering, University of Pavia, Via Ferrata 5, Pavia, 27100, Italy
        {\tt\small irene.schimperna01@universitadipavia.it, lalo.magni@unipv.it}}%
\thanks{$^{2}$Lea Bold and Karl Worthmann are with the Optimization-based Control Group, Institute of Mathematics, Technische Universit{\"a}t Ilmenau, Ilmenau, Germany
        {\tt\small lea.bold@tu-ilmenau.de, karl.worthmann@tu-ilmenau.de}}%
\thanks{$^{3}$Johannes Köhler is with the Institute for Dynamic Systems and Control, ETH Zurich,
Zurich 8053, Switzerland
        {\tt\small jkoehle@ethz.ch}}%
}

\begin{document}

\maketitle
\thispagestyle{empty}
\pagestyle{empty}

\begin{abstract}
    This paper derives conditions under which Model Predictive Control (MPC) with terminal conditions, using a data-driven surrogate model 
    as a prediction model, asymptotically stabilizes the plant despite approximation errors. 
    In particular, we prove recursive feasibility and asymptotic stability if a 
    proportional error bound holds, where proportional means that the bound is linear in the norm of the state and the input.
    For a broad class of nonlinear systems, this condition can be satisfied using data-driven surrogate models generated by kernel Extended Dynamic Mode Decomposition (kEDMD) using the Koopman operator. Last, the applicability of the proposed framework is demonstrated in a numerical case study.
\end{abstract}


\section{Introduction}

\noindent The use of data-driven surrogate models in Model Predictive Control (MPC) gained increasing popularity in the last years thanks to the development of many powerful learning methods~\cite{hewing2020learning}. 
MPC algorithms based on a variety of different data-driven models have been proposed, including nonlinear ARX models \cite{denicolao1997stabilizing}, Gaussian process regression~\cite{hewing2019cautious, koller2018learning} and neural networks \cite{ren2022tutorial,schimperna2024recurrent}. 
In this setting, a family of powerful methods are based on the Koopman operator and its data-driven approximation using Extended Dynamic Mode Decomposition (EDMD~\cite{williams2015data}), see~\cite{strasser2025overview} for a recent overview on Koopman-based control with closed-loop including a recap of the theoretical foundations and recent progress w.r.t.\ error bounds.
The Koopman operator is a linear but infinite-dimensional operator that encodes the behavior of the associated nonlinear dynamical system \cite{mezic2005spectral, rowley2009spectral, lazar2025product}. 
Different variants of EDMD-based surrogate models have been successfully employed in MPC. A popular choice is linear EDMD with control (EDMDc; \cite{korda2018linear}). 
An alternative are bilinear EDMDc models, that have better representation capabilities \cite{iacob:toth:schoukens:2022} and for which finite-data error bounds can be derived \cite{nuske2023finite,schaller2023towards}.
Finally, kernel EDMD (kEDMD; \cite{klus2020kernel}) is an appealing extension: On the one hand, the observables are chosen in a data-driven manner. On the other hand, recent pointwise error bound~\cite{kohne2025error} has been extended to systems with inputs~\cite{BoldPhil25} such that the approximation error can be rigorously analysed.

Standard stability results in MPC are derived under the assumption that the model exactly describes the system under control~\cite{RawlMayn17}. However, approximation errors are typically present in data-driven models. 
In this case, neither recursive feasibility nor stability are guaranteed. 
To ensure recursive feasibility in presence of modeling errors or disturbances it is necessary to 
use a robust MPC formulation, e.g., based on constraint tightening.
For linear systems, 
the tightening can be computed exactly \cite{kouvaritakis2016model}, while for nonlinear systems over-approximations are often used either based on the Lipschitz constant of the system under control~\cite{limon2002input} or on incremental stabilizability~\cite{kohler2018novel}.
Concerning closed-loop stability, a part of the literature focuses on inherent robustness of MPC, i.e., under which assumptions the MPC control law designed using the nominal model stabilizes also the uncertain system \cite{grimm2007nominally}.
Another popular method to achieve convergence to a desired reference despite model-plant mismatch is offset-free MPC~\cite{pannocchia2015offset, SchiBold25}.
For MPC without terminal conditions, the recent paper~\cite{schimperna2025data} has shown that asymptotic stability of the origin w.r.t.\ the MPC closed loop can be obtained if, in particular, proportional error bounds are available, i.e., when the model exactly describes the system at the origin and the size of the error is proportional to magnitude of the state and the input. 
This condition has been verified for kEDMD surrogate models in~\cite{BoldPhil25, schimperna2025data}.
Using similar conditions to~\cite{BoldPhil25} and~\cite{schimperna2025data}, \cite{kuntz2025beyond} have derived conditions under which a nominal MPC without state constraints but with terminal cost and terminal set ensures asymptotic stability in presence of a parametric model-plant mismatch. Finally, exponential stability of Koopman MPC with terminal ingredients is shown in \cite{shang2025exponential} using an EDMDc model, where, however, it is not studied how to guarantee the required bounds on the modeling error.

In this paper, we study conditions for asymptotic stability of MPC formulations with state constraints and terminal conditions based on data-driven surrogate models. 
In particular, we consider a terminal cost and a terminal set designed on the base of the surrogate model, and we introduce a constraint tightening to ensure recursive feasibility.
We prove that, in presence of sufficiently small
error bounds, the MPC designed with the data-driven surrogate model asymptotically stabilizes the unknown system. 
Moreover, we show how the requirement can be satisfied by kEDMD models in the Koopman framework. 
Finally, we verify our finding with numerical simulations.

On the one hand, our analysis extends the results presented in~\cite{schimperna2025data} to MPC with terminal conditions, allowing to use MPC with short horizons, and tightens 
the analysis of~\cite{WortStra24}, where only practical asymptotic stability was shown. 
On the other hand, we also extend the results proposed in~\cite{kuntz2025beyond} by proposing a learning framework such that the imposed assumptions can be rigorously verified to derive asymptotic stability. 
Since the required error bounds can only be ensured in the compact set where data are collected, we also had to incorporate state constraints in the MPC algorithm, which were not considered in~\cite{kuntz2025beyond}.

The paper is organized as follows. Section~\ref{sec:setup} introduces the system and the properties of the 
data-driven model. Section~\ref{sec:MPC:analysis} presents the proposed MPC algorithm and analyzes its stability properties. Section~\ref{sec:implementation} shows how the MPC stability conditions can be satisfied with kEDMD surrogate models. Simulation examples are reported in Section~\ref{sec:examples}, before conclusions are drawn in Section~\ref{sec:conclusion}.
\\

\noindent\textbf{Notation}: $\|\cdot\|$ is used for the Euclidean norm on $\R^n$ and for the induced matrix norm on $\R^{n \times n}$, while $\|\cdot\|_\mathcal{F}$ denotes the Frobenius norm. For a matrix $M$, we use the shorthand notation $\|\cdot\|_M^2 := x^\top M x$.
The pseudoinverse of a matrix $M$ is denoted by $M^\dagger$.
Given two numbers $a, b \in \mathbb{Z}$, the abbreviation $[a:b] := \mathbb{Z} \cap [a, b]$ is used. 
By $C_b(\Omega)$ the space of continuous bounded real-valued functions on the set $\Omega$ is denoted.
The ball of radius $r$ centered at~$x$ is denoted by $\mathcal{B}_r(x)$, while only using~$\mathcal{B}_r$ for $\mathcal{B}_r(0)$. 
The symbol $\ominus$ denotes the usual Pontryagin set difference. For $i \in \N$, $\delta_{i1}$ denotes the Kronecker-delta, i.e., $\delta_{i1} = 1$ for $i = 1$ and $0$ otherwise.

\section{Setup and MPC with terminal conditions} \label{sec:setup}

\noindent In this paper, we consider the discrete-time control system 
\begin{equation}\label{eq:sys}
    x^+ = f(x,u) 
\end{equation}
with state $x \in \mathbb{S}$, input $u \in \mathbb{U}$ and successor 
state $x^+ \in \R^n$, where $\mathbb{S} \subset \R^n$ and $\mathbb{U} \subset \R^m$ are compact, convex sets containing the origin in their interior.
Let the map $f: \mathbb{S} \times \mathbb{U} \to \R^n$ be continuous and locally Lipschitz continuous w.r.t.\ its first argument.
We assume that the origin is a controlled equilibrium of the system, i.e., $f(0,0) = 0$.

Since the system dynamics~\eqref{eq:sys} are considered to be unknown, the MPC controller uses the data-driven surrogate
\begin{equation} \label{eq:sys:surrogate}
    x^+ = f^\varepsilon(x,u)
\end{equation}
parametrized in $\varepsilon \in (0, \bar{\varepsilon}]$, $\bar{\varepsilon} \in (0, \infty)$, in the optimization step. 
The superscript~$\varepsilon$ in~$f^\varepsilon$ refers to the approximation accuracy, and is used in the paper to refer to all the parameters related to the surrogate dynamics~\eqref{eq:sys:surrogate}.

The goal of the paper is to show that, if the surrogate~\eqref{eq:sys:surrogate} is sufficiently accurate, the MPC algorithm ensures constraint satisfaction and asymptotic stability when applied to the system~\eqref{eq:sys}. 
In particular, we assume the existence of proportional and uniform error bounds for the data-driven surrogate model similarly to~\cite{bold2024data,WortStra24} and, later also~\cite{BoldScha25,kuntz2025beyond,schimperna2025data}.
\begin{assumption} \label{ass:error_bounds}
    Let the following hold for all $\varepsilon \in (0,\bar{\varepsilon}]$:
    \begin{enumerate}
        \item The proportional and uniform error bound 
        \begin{equation}\label{eq:ass:error_bound:proportional}
            \|f(x, u) - f^\varepsilon(x, u)\| \leq \min\{ c_x^\varepsilon\|x\| + c_u^\varepsilon\|u\|, \eta^\varepsilon \}
        \end{equation}
        holds for all $x \in \mathbb{S}, u \in \mathbb{U}$ with parameters $c_x^\varepsilon, c_u^\varepsilon, \eta^\varepsilon$ satisfying $\lim_{\varepsilon \searrow 0} \max \{c_x^\varepsilon, c_u^\varepsilon, \eta^\varepsilon\} = 0$. 
        \item The surrogate dynamics~\eqref{eq:sys:surrogate} is Lipschitz continuous on the set $\mathbb{S} \times \mathbb{U}$, i.e., $\exists\,\bar{L} \geq 0$ such that for every $u \in \mathbb{U}$
        \begin{equation}\label{eq:ass:Lipschitz}
            \|f^\varepsilon(x,u) - f^\varepsilon(y,u)\| \leq 
            \bar{L} \|x-y\| \quad\forall\,x,y \in \mathbb{S}.
        \end{equation}
    \end{enumerate}
\end{assumption}
Section~\ref{sec:implementation} recaps a method to generate data-driven surrogate models satisfying Assumption~\ref{ass:error_bounds} using the Koopman operator.

We consider a standard MPC formulation, where stability and recursive feasibility are ensured through terminal conditions.
The MPC considers quadratic stage cost
\begin{equation*}
    \ell(x, u) = \|x\|_Q^2 + \|u\|_R^2 = x^\top Q x + u^\top R u
\end{equation*}
with symmetric and positive definite weighting matrices $Q \in \R^{n \times n}$ and $R \in \R^{m \times m}$ as well as a terminal cost $V_\mathrm{f}^\varepsilon$ and a terminal region $\mathbb{X}_\mathrm{f}^\varepsilon$ designed using the surrogate model. 

In the following, we introduce the definition of admissibility of a sequence $\mathbf{u} = (u(k))_{k=0}^{N-1} \subset \mathbb{U}$ for an initial state~$\hat{x}$ and a prediction horizon $N \in \mathbb{N}$, where a constraint tightening is incorporated to ensure recursive feasibility despite modeling errors.
To do so, let the state trajectory associated to $\mathbf{u}$ be given by $(x^\varepsilon_{\mathbf{u}}(k; \hat{x}))_{k=0}^{N}$, where $x^\varepsilon_{\mathbf{u}}(0; \hat{x}) = \hat{x}$ and $x^\varepsilon_{\mathbf{u}}(k; \hat{x}) = f^\varepsilon(x^\varepsilon_{\mathbf{u}}(k-1; \hat{x}), u(k-1))$ for $k \in [1:N]$ .
\begin{definition} [$\varepsilon$-admissibility] 
\label{def:admissibility}
    Consider the surrogate model~\eqref{eq:sys:surrogate} with $\varepsilon \in (0,\bar{\varepsilon}]$.
    $\mathbf{u}$ is said to be an $\varepsilon$-admissible control sequence for state $\hat{x} \in S$ and horizon $N \in \mathbb{N}$ if 
    \begin{equation*}
        x^\varepsilon_{\mathbf{u}}(k; \hat{x}) \in \begin{cases}
            \mathbb{S} \ominus \mathcal{B}_{\bar{c}(k) \eta^\varepsilon} & \text{for all $k \in [1:N-1]$} \\
            \mathbb{X}_\mathrm{f}^\varepsilon & \text{for $k = N$}
        \end{cases}
    \end{equation*}
    with $\bar{c}(k) = \sum_{i=0}^{k-1} \bar{L}^i$, $\bar{L}$ defined in Assumption~\ref{ass:error_bounds}.
    The set of $\varepsilon$-admissible control sequences is denoted by $\mathcal{U}_N^\varepsilon(\hat{x})$.   
\end{definition}

The admissibility definition considers tightened constraints based on the Lipschitz constant $\bar{L}$ and the uniform error bound $\eta^\varepsilon$ as originally proposed in~\cite{limon2002input}.
A less conservative constraint tightening approach using incremental stabilizability can be found in~\cite{kohler2018novel} and is of interest for future work. 

Given the admissibility definition, it is possible to state the MPC algorithm, which is reported in Algorithm \ref{alg:MPC}. 
\begin{algorithm}[htb]
    \caption{MPC with terminal conditions using the surrogate model~\eqref{eq:sys:surrogate} in the optimization step}\label{alg:MPC}
    \raggedright
    \hrule
    \smallskip
    {\it Input:} Horizon $N \in \N$, stage cost~$\ell$,
    surrogate~$f^\varepsilon$, state constraint set~$\mathbb{S}$, terminal cost~$V_\mathrm{f}^\varepsilon$, terminal region~$\mathbb{X}_\mathrm{f}^\varepsilon$, control constraint set $\mathbb{U}$.
    \smallskip\hrule
    \medskip
    \textit{Initialization}: Set $k = 0$.\\[2mm]
    \noindent\textit{(1)} Measure current state $x_{\mu_N^\varepsilon}(k)$ and set $\hat{x} = x_{\mu_N^\varepsilon}(k)$.\\[1mm]
    \noindent\textit{(2)} Solve the optimal control problem~
    \begin{align}\label{eq:OCP}\tag{OCP}
    \begin{split}
        \min_{\mathbf{u} \in \mathcal{U}_N^\varepsilon(\hat{x})} \;\; & J_N^\varepsilon(\hat{x}, \mathbf{u}) := \sum_{i = 0}^{N-1} \! \ell(x_\mathbf{u}^\varepsilon(i;\hat{x}), u(i)) + V_\mathrm{f}^\varepsilon(x_\mathbf{u}^\varepsilon(N;\hat{x})) \\
        \text{s.t.} \;\; & x_\mathbf{u}^\varepsilon(0;\hat{x}) = \hat{x} \\
        & x_\mathbf{u}^\varepsilon(i+1;\hat{x}) \! = f^\varepsilon(x_\mathbf{u}^\varepsilon(i;\hat{x}), u(i)), \;\;\; i \in [0:N-1]
    \end{split}
    \end{align}
    \hspace*{5mm} to obtain an
    optimal input sequence~$\mathbf{u}^\star = (u^\star(i))_{i = 0}^{N - 1}$. \\[1mm]
    \noindent\textit{(3)} Apply the MPC feedback law~$\mu_N^\varepsilon(\hat{x}) = u^\star({0})$ at the\\
    \hspace*{5mm} plant \eqref{eq:sys}, increment~$k$ and go to Step~\textit{(1)}.
    \smallskip\hrule
\end{algorithm}

We define the optimal value function as
\begin{equation*}
    V_N^\varepsilon (\hat{x}) := \inf_{\mathbf{u} \in \mathcal{U}_N^\varepsilon (\hat{x})} J_N^\varepsilon (\hat{x}, \mathbf{u}),
\end{equation*}
which will be used in the proof of asymptotic stability. For simplicity of exposition, we assume that a minimizer exists, see, e.g., \cite{grune2017nonlinear} for remedies.
To prove recursive feasibility and asymptotic stability, the terminal cost function~$V_\mathrm{f}^\varepsilon$ and the terminal region~$\mathbb{X}_\mathrm{f}^\varepsilon$ must satisfy the following assumption.
\begin{assumption} [Terminal conditions] \label{ass:terminal}
    For $\varepsilon \in (0,\bar{\varepsilon}]$, the terminal cost~$V_\mathrm{f}^\varepsilon$ can be written as
    \begin{equation} \label{eq:structure-Vf}
        V_\mathrm{f}^\varepsilon(x) = \| \Phi^\varepsilon(x) \|_P^2 = \Phi^\varepsilon(x)^\top P \Phi^\varepsilon(x)
    \end{equation}
    for a positive definite matrix $P$ and a Lipschitz-continuous function~$\Phi^\varepsilon: \mathbb{S} \rightarrow \mathbb{R}^M$ with Lipschitz constant $L_\Phi$ satisfying $\Phi^\varepsilon(0) = 0$, and the terminal set is a sublevel set of the terminal cost, i.e.
    \begin{equation} \label{eq:structure-Xf}
        \mathbb{X}_\mathrm{f}^\varepsilon = \{ x \in \R^n : V_\mathrm{f}^\varepsilon(x) \leq c \}
    \end{equation}
    with $c > 0$, and is such that $\mathbb{X}_\mathrm{f}^\varepsilon \subseteq \mathbb{S} \ominus \mathcal{B}_{\bar{c}(N) \eta^\varepsilon}$.
    Moreover, there exists a control law $\mu^\varepsilon: \mathbb{X}_\mathrm{f}^\varepsilon \to \mathbb{U}$ with $\mu^\varepsilon(0) = 0$ such that the Lyapunov decrease condition
    \begin{equation}\label{eq:terminal-cost-eps}
        V_\mathrm{f}^\varepsilon(f^\varepsilon(x, \mu^\varepsilon(x))) - V_\mathrm{f}^\varepsilon(x) \leq -\ell(x, \mu^\varepsilon(x))
    \end{equation}
    holds for all $x \in \mathbb{X}_\mathrm{f}^\varepsilon$.
\end{assumption}

A design of $V_\mathrm{f}^\varepsilon$ and $\mathbb{X}_\mathrm{f}^\varepsilon$ satisfying Assumption~\ref{ass:terminal} based on a linearization of the surrogate model at the origin is proposed at the end of Section~\ref{sec:MPC:analysis}.

\section{Recursive feasibility \& asymptotic stability}
\label{sec:MPC:analysis}

In the following, we present the main result of the paper, where we show that in presence of sufficiently tight proportional and uniform error bounds, see Property 1 of Assumption~\ref{ass:error_bounds}, the MPC using the surrogate model~\eqref{eq:sys:surrogate} ensures asymptotic stability of the original system~\eqref{eq:sys}.
\begin{theorem} \label{thm:AS_MPC}
    Let Assumptions~\ref{ass:error_bounds} and~\ref{ass:terminal} hold for a sufficiently small $\varepsilon > 0$ and assume initial feasibility, i.e.,  
    $\mathcal{U}^\varepsilon_N(x_{\mu_N^\varepsilon}(0)) \neq \emptyset$.
    Then, the optimization problem of Algorithm~\ref{alg:MPC} is 
    feasible for all $k \in \mathbb{Z}_{\geq 0}$, and the origin is asymptotically stable for the MPC closed loop. 
\end{theorem}
\begin{proof}
    \textbf{Recursive feasibility.}
    Let $\hat{x} \in \mathbb{S}$ be a feasible initial state with $\mathcal{U}_N^\varepsilon(\hat{x}) \neq \emptyset$. 
    Let $\mathbf{u}^\star := (u^\star_i)_{i=0}^{N-1}$ and $(x^\varepsilon_{\mathbf{u}^\star}(i))_{i=0}^N$ be the optimal control sequence computed in the MPC optimization step and the corresponding (predicted) state trajectory. 
    Recursive feasibility can be shown as in~\cite{WortStra24} using the control sequence
    \begin{equation} \label{eq:feas-seq}
        \mathbf{u}^+ := \left( u^\star_1, \ldots, u^\star_{N-1}, \mu^\varepsilon(x^\varepsilon_{\mathbf{u}^\star}(N)) \right),
    \end{equation}
    which is the tail of the optimal control sequence prolonged by an additional element using the terminal controller~$\mu^\varepsilon$.
    The corresponding state trajectory is denoted by $(x^\varepsilon_{\mathbf{u}^+}(i))_{i=0}^{N}$, where $x^\varepsilon_{\mathbf{u}^+}(0) = x^+ := f(\hat{x}, \mu_N^\varepsilon(\hat{x})) = f(\hat{x},u^\star_0)$. 
    Note that, in general, $x^+ \neq f^\varepsilon(\hat{x},u^\star_0) = x^\varepsilon_{\mathbf{u}^\star}(1;\hat{x})$ meaning that the predicted successor state may deviate from its counterpart at the plant. 
    We also prolong the optimal sequence~$\mathbf{u}^\star$ using the last element of $\mathbf{u}^+$, i.e., we define $u^\star_N := \mu^\varepsilon(x^\varepsilon_{\mathbf{u}^\star}(N))$. The corresponding prolonged state sequence includes as last element $x^\varepsilon_{\mathbf{u}^\star}(N+1) := f^\varepsilon(x^\varepsilon_{\mathbf{u}^\star}(N), u^\star_N)$.
    Then, using Inequality~\eqref{eq:ass:error_bound:proportional} of Assumption~\ref{ass:error_bounds}, the difference between $(x^\varepsilon_{\mathbf{u}^\star}(i))_{i=1}^{N+1} = (x^\varepsilon_{\mathbf{u}^\star}(i;\hat{x}))_{i=1}^{N+1}$ and $(x^\varepsilon_{\mathbf{u}^+}(i))_{i=0}^{N} = (x^\varepsilon_{\mathbf{u}^+}(i;x^+))_{i=0}^{N}$ can be bounded in the following way:
    \begin{equation*}
        \| x^\varepsilon_{\mathbf{u}^+}(0) - x^\varepsilon_{\mathbf{u}^\star}(1)\| = \|f(\hat{x}, u^\star_0) - f^\varepsilon(\hat{x}, u^\star_0)\| \leq \eta^\varepsilon,
    \end{equation*}
    \begin{equation} \label{eq:error-x-i}
        \| x^\varepsilon_{\mathbf{u}^+}(i) - x^\varepsilon_{\mathbf{u}^\star}(i+1)\| \leq \bar{L}^i \eta^\varepsilon.
    \end{equation}
    Then, since $\bar{c}(i+1) = \bar{L}^i + \bar{c}(i)$, we have the implication
    \begin{equation*}
        x^\varepsilon_{\mathbf{u}^\star}(i+1) \in \mathbb{S} \ominus \mathcal{B}_{\bar{c}(i+1) \eta^\varepsilon} \implies x^\varepsilon_{\mathbf{u}^+}(i) \in \mathbb{S} \ominus \mathcal{B}_{\bar{c}(i) \eta^\varepsilon}
    \end{equation*}
    showing that states $(x^\varepsilon_{\mathbf{u}^+}(i))_{i=0}^{N-1}$ respect the tightened constraints. 
    
    Lastly, we show feasibility of terminal constraint, i.e., we show that $x^\varepsilon_{\mathbf{u}^+}(N) \in \mathbb{X}_\mathrm{f}^\varepsilon$. Let $c_1 = c_1(\eta^\varepsilon) \in (0,c)$ be the largest number such that $V_\mathrm{f}^{-1}[0, c_1] \subseteq \mathbb{X}_\mathrm{f}^\varepsilon \ominus \mathcal{B}_{\bar{L}^N \eta^\varepsilon}$, where $V_\mathrm{f}^{-1}[0, c_1]$ is the sublevel set $\{ x : V_\mathrm{f}^\varepsilon(x) \leq c_1 \}$. 
    In view of Assumption \ref{ass:terminal}, it holds that
    \begin{equation*}
        \min_{x \in \mathbb{X}_\mathrm{f}^\varepsilon \setminus V_\mathrm{f}^{-1}[0, c_1]} V_\mathrm{f}^\varepsilon(x) - V_\mathrm{f}^\varepsilon(f^\varepsilon(x, \mu^\varepsilon(x))) > 0.
    \end{equation*}
    Hence, for sufficiently small $\varepsilon$, and, thus, a sufficiently small $\eta^\varepsilon$, we have that $V_\mathrm{f}^\varepsilon(x_{\mathbf{u}^\star}(N+1)) \leq c_1$. 
    In fact, if $x^\varepsilon_{\mathbf{u}^\star}(N) \in V_\mathrm{f}^{-1}[0, c_1]$ then also $x^\varepsilon_{\mathbf{u}^\star}(N+1)$ is in the same set because sublevel sets of $V_\mathrm{f}^\varepsilon$ are forward invariant under the control law $\mu^\varepsilon$. In case $x^\varepsilon_{\mathbf{u}^\star}(N) \notin V_\mathrm{f}^{-1}[0, c_1]$, then $V_\mathrm{f}^\varepsilon$ must have a positive decrease, and thus if $c_1$ is sufficiently close to $c$ then $x^\varepsilon_{\mathbf{u}^\star}(N+1)$ is in the set $V_\mathrm{f}^{-1}[0, c_1]$. In view of the choice of $c_1$, we have that $x^\varepsilon_{\mathbf{u}^\star}(N+1) \in \mathbb{X}_\mathrm{f}^\varepsilon \ominus \mathcal{B}_{\bar{L}^N \eta^\varepsilon}$, which implies that $x^\varepsilon_{\mathbf{u}^+}(N) \in \mathbb{X}_\mathrm{f}^\varepsilon$ using the bound in \eqref{eq:error-x-i}.
    \\
    
    \noindent \textbf{Asymptotic stability}. We consider the optimal value function $V_N^\varepsilon$ as a candidate Lyapunov function. 
    Invoking optimality of the computed control sequence~$\mathbf{u}^\star = \mathbf{u}^\star(\cdot;\hat{x})$, it follows that
    \begin{subequations}
    \begin{align}
        & V_N^\varepsilon(x^+) - V_N^\varepsilon(\hat{x}) \leq J_N^\varepsilon(x^+, \mathbf{u}^+) - V_N^\varepsilon(\hat{x}) \nonumber \\
        \leq& -\ell(x^\varepsilon_{\mathbf{u}^\star}(0), u^\star_0) \nonumber \\
        & + \sum_{i=0}^{N-2} \Big[ \ell(x^\varepsilon_{\mathbf{u}^+}(i), u^+_i) - \ell(x^\varepsilon_{\mathbf{u}^\star}(i+1), u^\star_{i+1}) \Big] \label{eq:decrease-1} \\
        & + \ell(x^\varepsilon_{\mathbf{u}^+}(N-1), u^+_{N-1}) + V_\mathrm{f}^\varepsilon(x^\varepsilon_{\mathbf{u}^+}(N)) - V_\mathrm{f}^\varepsilon(x^\varepsilon_{\mathbf{u}^\star}(N)) \label{eq:decrease-2}
    \end{align}
    \end{subequations}
    Next, we consider the terms~\eqref{eq:decrease-1} and~\eqref{eq:decrease-2} separately, beginning with the term~\eqref{eq:decrease-1}. Here, we use the fact that $u^+_i = u^\star_{i+1}$ for all $i \in [0: N-2]$. Hence, we have for the $i$-th summand
    \begin{align} 
        & \ell(x^\varepsilon_{\mathbf{u}^+}(i), u^+_i) - \ell(x^\varepsilon_{\mathbf{u}^\star}(i+1), u^\star_{i+1}) \label{eq:proof:MPC:stability:estimate} \\
        = \;& \|x^\varepsilon_{\mathbf{u}^+}(i) \|_Q^2 - \|x^\varepsilon_{\mathbf{u}^\star}(i+1) \|_Q^2 \nonumber \\
        \leq \;& \|Q\| \underbrace{\| x^\varepsilon_{\mathbf{u}^+}(i) + x^\varepsilon_{\mathbf{u}^\star}(i+1) \|}_{= \| 2x^\varepsilon_{\mathbf{u}^\star}(i+1) + x^\varepsilon_{\mathbf{u}^+}(i) - x^\varepsilon_{\mathbf{u}^\star}(i+1)\| } \|x^\varepsilon_{\mathbf{u}^+}(i) - x^\varepsilon_{\mathbf{u}^\star}(i+1)\| \nonumber \\
        \leq \;& 2 \|Q\| \| x^\varepsilon_{\mathbf{u}^\star}(i+1)\| \|x^\varepsilon_{\mathbf{u}^+}(i) - x^\varepsilon_{\mathbf{u}^\star}(i+1)\| \nonumber \\
        & + \|Q\| \| x^\varepsilon_{\mathbf{u}^+}(i) - x^\varepsilon_{\mathbf{u}^\star}(i+1)\|^2 \nonumber
    \end{align}
    where we have used $\|a\|_Q^2 - \|b\|_Q^2 = (a+b)^\top Q (a-b) \leq \|Q\| \|a+b\| \|a-b\|$ in the first inequality.
    Then, using an analogous reasoning as used in the paragraph on recursive feasibility, we get
    \begin{eqnarray}\nonumber
        \| x^\varepsilon_{\mathbf{u}^+}(i) - x^\varepsilon_{\mathbf{u}^\star}(i+1)\| & \leq & \bar{L}^i \| x^\varepsilon_{\mathbf{u}^+}(0) - x^\varepsilon_{\mathbf{u}^\star}(1) \| \\
        & \leq & \bar{L}^i (c_x^\varepsilon \|\hat{x}\| + c_u^\varepsilon \|u^\star_0\|), \label{eq:thm:stability:proof2}
    \end{eqnarray}
    using $x^\varepsilon_{\mathbf{u}^+}(0) - x^\varepsilon_{\mathbf{u}^\star}(1) = x^+ - f^\varepsilon(\hat{x}, u^\star_0) = f(\hat{x}, u^\star_0) - f^\varepsilon(\hat{x}, u^\star_0)$.
    
    Next, we consider the term~\eqref{eq:decrease-2} and extend it suitably, i.e.,
    \begin{align}
        & \ell(x^\varepsilon_{\mathbf{u}^+}(N-1), u^+_{N-1}) \pm \ell(x^\varepsilon_{\mathbf{u}^\star}(N), \mu^\varepsilon(x_{\mathbf{u}^\star}(N))) \nonumber\\
        & + V_\mathrm{f}^\varepsilon(x^\varepsilon_{\mathbf{u}^+}(N)) \pm 
        V_\mathrm{f}^\varepsilon(x^\varepsilon_{\mathbf{u}^\star}(N+1))
        - V_\mathrm{f}^\varepsilon(x^\varepsilon_{\mathbf{u}^\star}(N)) \nonumber
    \end{align}
    The difference $\ell(x^\varepsilon_{\mathbf{u}^+}(N-1), u^+_{N-1}) - \ell(x^\varepsilon_{\mathbf{u}^\star}(N), \mu^\varepsilon(x^\varepsilon_{\mathbf{u}^\star}(N)))$ can be treated analogously to the summands in~\eqref{eq:decrease-1} and will later be included in the respective summation ($i = N-1$).
    Moreover, we have
    \begin{align*}
        & \ell(x^\varepsilon_{\mathbf{u}^\star}(N), \mu^\varepsilon(x^\varepsilon_{\mathbf{u}^\star}(N))) + 
        V_\mathrm{f}^\varepsilon(x^\varepsilon_{\mathbf{u}^\star}(N+1)) \\
        & - V_\mathrm{f}^\varepsilon(x^\varepsilon_{\mathbf{u}^\star}(N)) \leq 0    
    \end{align*}
    in view of Assumption~\ref{ass:terminal}. Then, leveraging $V_\mathrm{f}^\varepsilon(x) = \| \Phi^\varepsilon(x) \|_P^2$, we can estimate the remaining difference analogously to Inequality~\eqref{eq:proof:MPC:stability:estimate} to obtain
    \begin{align*}
        & V_f(x^\varepsilon_{\mathbf{u}^+}(N)) - V_f(x^\varepsilon_{\mathbf{u}^\star}(N+1)) \\
        \leq & 2 \|P\| \| \Phi^\varepsilon\!(x^\varepsilon_{\mathbf{u}^\star}\!(N+1))\|  \|\Phi^\varepsilon\!(x^\varepsilon_{\mathbf{u}^+}(N)) - \Phi^\varepsilon\!(x^\varepsilon_{\mathbf{u}^\star}\!(N+1))\| \\
        & +  \|P\| \|\Phi^\varepsilon(x^\varepsilon_{\mathbf{u}^+}(N)) - \Phi^\varepsilon(x^\varepsilon_{\mathbf{u}^\star}(N+1))\|^2.
    \end{align*}
    The term $\|\Phi^\varepsilon(x^\varepsilon_{\mathbf{u}^+}(N)) - \Phi^\varepsilon(x^\varepsilon_{\mathbf{u}^\star}(N+1))\|$ is treated similarly to~\eqref{eq:thm:stability:proof2}, which yields the estimate
    \begin{align*}
        & \|\Phi^\varepsilon(x^\varepsilon_{\mathbf{u}^+}(N)) - \Phi^\varepsilon(x^\varepsilon_{\mathbf{u}^\star}(N+1))\| \\
        \leq \;& L_\Phi \|x^\varepsilon_{\mathbf{u}^+}(N) - x^\varepsilon_{\mathbf{u}^\star}(N+1)\| \\
        \leq \;& L_\Phi \bar{L}^N (c_x^\varepsilon \|\hat{x}\| + c_u^\varepsilon \|u_0^\star\|).
    \end{align*}
    
    Overall, using the derived bounds on the terms~\eqref{eq:decrease-1} and~\eqref{eq:decrease-2}, we get
    \begin{align} \label{eq:decrease-V-PAS}
        & V_N^\varepsilon(x^+) - V_N^\varepsilon(\hat{x}) \leq -\ell(x^\varepsilon_{\mathbf{u}^\star}(0), u^\star_0) \\
        &\;\;\; + 2 (c_x^\varepsilon \|\hat{x}\| + c_u^\varepsilon \|u^\star_0\|) \Bigg[ \|Q\| \sum_{i=0}^{N-1} \bar{L}^i \| x^\varepsilon_{\mathbf{u}^\star}(i+1)\| \nonumber \\
        &\;\;\; + \|P\| \| \Phi(x^\varepsilon_{\mathbf{u}^\star}(N+1))\| L_\Phi \bar{L}^N \Bigg] \nonumber \\ 
        &\;\;\; + (c_x^\varepsilon \|\hat{x}\| + c_u^\varepsilon \|u^\star_0\|)^2 \left[ \|Q\| \sum_{i=0}^{N-1} \bar{L}^{2i} + \|P\| L_\Phi^2 \bar{L}^{2N} \right] \nonumber
    \end{align}
    To prove asymptotic stability, we have to show a Lyapunov decrease in~\eqref{eq:decrease-V-PAS}. To this end, we mainly have to show that
    both terms are proportional to~$\| \hat{x} \|^2 + \|u^\star_0\|^2$. 

    First, consider the case $\hat{x} \in \mathbb{X}_\mathrm{f}^\varepsilon$.
    In view of Assumption~\ref{ass:terminal}, it is possible to apply the control law $\mu^\varepsilon$ iteratively for $N$ steps to $\hat{x}$.
    Then, due to optimality of the solution in the MPC optimization step solution, we have
    \begin{align*} 
        V_N^\varepsilon(\hat{x}) 
        & \leq \! \sum_{i=0}^{N-1} \! \left( \| x^\varepsilon_{\mu^\varepsilon}(i)\|_Q^2 + \| \mu^\varepsilon(x^\varepsilon_{\mu^\varepsilon}(i)) \|_R^2 \right) + V_\mathrm{f}^\varepsilon(x^\varepsilon_{\mu^\varepsilon}(N)) \nonumber \\
        & \stackrel{\eqref{eq:terminal-cost-eps}}{\leq} V_\mathrm{f}^\varepsilon(\hat{x}) \stackrel{\eqref{eq:structure-Vf}}{\leq} \| \Phi^\varepsilon(\hat{x}) \|_P^2 \leq \lambda_{\max}(P) L_\Phi^2 \| \hat{x} \|^2,
    \end{align*}
    where $\lambda_{\max}(P) > 0$ is the positive definite maximum eigenvalue of matrix $P$, and $x^\varepsilon_{\mu^\varepsilon}(i)$ is the $i$-th element of the state sequence obtained applying the control law $\mu^\varepsilon$ starting from $\hat{x}$.
    The first inequality in the second line is obtained by applying~\eqref{eq:terminal-cost-eps} iteratively in a telescopic sum.
    Since we consider compact constraint sets, a finite optimization horizon $N$ and a continuous cost function, and since the terminal set $\mathbb{X}^\varepsilon_\mathrm{f}$ contains the origin in its interior, it holds that
    \begin{align} \label{eq:upper-bound-with-Vf}
        V_N^\varepsilon(\hat{x}) & = \sum_{i=0}^{N-1} \left( \| x^\varepsilon_{\mathbf{u}^\star}(i)\|_Q^2 + \|u^\star_i\|_R^2 \right) + V_\mathrm{f}^\varepsilon(x^\varepsilon_{\mathbf{u}^\star}(N)) \nonumber \\
        & \leq c_2 \| \hat{x} \|^2
    \end{align}
    for all feasible initial states $\hat{x} \in \mathbb{S}$, with $c_2 \geq \lambda_{\max}(P) L_\Phi^2$, see also \cite[Prop. 2.38]{RawlMayn17}.
    
    Hence, for each $i \in [0:N-1]$, we have
    \begin{align*}
        \| x^\varepsilon_{\mathbf{u}^\star}(i)\|_Q^2 \leq V_N^\varepsilon(\hat{x}) \leq c_2 \| \hat{x} \|^2.
    \end{align*}
    Using the fact that $\lambda_{\min}(Q) \| x^\varepsilon_{\mathbf{u}^\star}(i)\|^2 \leq \| x^\varepsilon_{\mathbf{u}^\star}(i)\|_Q^2$ and by taking the square root of previous expression, we obtain that for all $i \in [0:N-1]$
    \begin{align} \label{eq:proof:bound-x1}
        \| x^\varepsilon_{\mathbf{u}^\star}(i)\| \leq \Tilde{c}  \|\hat{x}\|,
    \end{align}
    with $\Tilde{c} := \sqrt{\frac{c_2}{\lambda_{\min}(Q)}}$.
    
    Similarly, considering the terms $\| x^\varepsilon_{\mathbf{u}^\star}(N)\|$ and $\| \Phi(f^\varepsilon(x^\varepsilon_{\mathbf{u}^\star}(N), \mu^\varepsilon(x^\varepsilon_{\mathbf{u}^\star}(N))))\|$, and
    applying~\eqref{eq:terminal-cost-eps} and~\eqref{eq:upper-bound-with-Vf}, we derive that
    \begin{align*}
        &\| \Phi^\varepsilon(x^\varepsilon_{\mathbf{u}^\star}(N+1)) \|_P^2 + \| x^\varepsilon_{\mathbf{u}^\star}(N)\|_Q^2 \\
        \leq & V_\mathrm{f}^\varepsilon(x^\varepsilon_{\mathbf{u}^\star}(N+1)) + \ell(x^\varepsilon_{\mathbf{u}^\star}(N), \mu^\varepsilon(x^\varepsilon_{\mathbf{u}^\star}(N))) \\
        \stackrel{\eqref{eq:terminal-cost-eps}}{\leq}& V_\mathrm{f}^\varepsilon(x^\varepsilon_{\mathbf{u}^\star}(N)) \stackrel{\eqref{eq:upper-bound-with-Vf}}{\leq} c_2 \| \hat{x} \|^2.
    \end{align*}
    Considering one term at a time on the left hand side and taking the square root, we obtain that the bound \eqref{eq:proof:bound-x1} holds also for $i=N$, and that
    \begin{equation} \label{eq:proof:bound-x3}
        \| \Phi^\varepsilon (x^\varepsilon_{\mathbf{u}^\star}(N+1)) \| \leq
        \Tilde{c} \|\hat{x}\|.
    \end{equation}
    
    Substituting bounds \eqref{eq:proof:bound-x1} 
    and \eqref{eq:proof:bound-x3} in \eqref{eq:decrease-V-PAS} and using the fact that $\|\hat{x}\| \|u^\star_0\| \leq \frac{1}{2}(\|\hat{x}\|^2 + \|u^\star_0\|^2)$ and that $(\|\hat{x}\| + \|u^\star_0\|)^2 \leq 2\|\hat{x}\|^2 + 2\|u^\star_0\|^2$, we get
    \begin{equation*}
        V_N^\varepsilon(x^+) - V_N^\varepsilon(\hat{x}) \leq -\ell(x^\varepsilon_{\mathbf{u}^\star}(0), u^\star_0) + C_x \|\hat{x}\|^2 + C_u \|u^\star_0\|^2
    \end{equation*}
    where
    \begin{align*}
        C_x :=& \sum_{i=0}^{N-1} \Big[ \|Q\| \Tilde{c} \bar{L}^i (2 c_x^\varepsilon + c_u^\varepsilon) + 2 \|Q\| \bar{L}^{2i} (c_x^\varepsilon)^2 \Big] \\
        &+ \|P\| \Tilde{c} L_\Phi \bar{L}^N (2 c_x^\varepsilon + c_u^\varepsilon) + 2 \|P\| L_\Phi^2 \bar{L}^{2N} (c_x^\varepsilon)^2,
    \end{align*}
    \begin{align*}
        C_u :=& \sum_{i=0}^{N-1} \Big[ \|Q\| \Tilde{c} \bar{L}^i c_u^\varepsilon + 2 \|Q\| \bar{L}^{2i} (c_u^\varepsilon)^2 \Big] \\
        &+ \|P\| \Tilde{c}L_\Phi \bar{L}^N c_u^\varepsilon + 2 \|P\| L_\Phi^2 \bar{L}^{2N} (c_u^\varepsilon)^2.
    \end{align*}
    Then, there exists a sufficiently small $\varepsilon > 0$ for which 
    \begin{equation*}
        -\ell(x^\varepsilon_{\mathbf{u}^\star}(0), u^\star_0) + C_x \|\hat{x}\|^2 + C_u \|u^\star_0\|^2 \leq - \alpha \|\hat{x}\|^2
    \end{equation*}
    for some $\alpha > 0$. 
    Combining the decrease of $V_N^\varepsilon$ with quadratic lower and upper bounds (see \eqref{eq:upper-bound-with-Vf}), we prove exponential stability of the closed loop.
\end{proof}

Theorem~\ref{thm:AS_MPC} shows that recursive feasibility and asymptotic stability can be achieved in presence of sufficiently tight error bounds. Uniform error bounds are required to define the constraint tightening and for the proof of recursive feasibility. Instead, the proportionality of the error bounds is the key requirement to obtain asymptotic stability. 
In fact, in presence of proportional error bounds, the prediction error decreases when the state approaches the origin, guaranteeing a decrease of the optimal value function for all the feasible states. 
The exact constant $\varepsilon > 0$ for which asymptotic stability is proven in Theorem~\ref{thm:AS_MPC} is conservative, and the result should be interpreted in a more qualitative way, i.e., asymptotic stability can be achieved for sufficiently tight error bounds.
The techniques used in the proof of Theorem~\ref{thm:AS_MPC} resemble the ones used in \cite{schimperna2025data}, but are adapted to consider the terminal conditions. In this case, asymptotic stability can be achieved for any value of the prediction horizon $N$.

A standard method for the design of terminal conditions satisfying Assumption~\ref{ass:terminal} is based on a linearization of the surrogate at the origin, see \cite[Section 2.5.5]{RawlMayn17}, given by 
\begin{equation*}
    A^\varepsilon := \frac{\partial f^\varepsilon}{\partial x}(0,0), \quad B^\varepsilon := \frac{\partial f^\varepsilon}{\partial u}(0,0).
\end{equation*}
If the pair $(A^\varepsilon, B^\varepsilon)$ is stabilizable,
it is possible to use the control law $\mu^\varepsilon(x) = K x$ for the design of $V_\mathrm{f}^\varepsilon$ and $\mathbb{X}_\mathrm{f}^\varepsilon$, where $K \in \R^{m \times n}$ is such that the matrix $A^\varepsilon + B^\varepsilon K$ is Schur. 
Then, a choice of the terminal cost satisfying Assumption \ref{ass:terminal} is the quadratic function
\begin{equation*}
    V_\mathrm{f}^\varepsilon (x) = \| \Phi^\varepsilon(x)\|_P^2,
\end{equation*}
where $\Phi^\varepsilon$ is the identity function, i.e. $\Phi^\varepsilon(x) = x$, and the matrix $P$ is such that
\begin{equation*}
    (A^\varepsilon + B^\varepsilon K)^\top P (A^\varepsilon + B^\varepsilon K) - P \preceq - \beta (Q + K^\top R K)
\end{equation*}
where $\beta > 1$ accounts for the local linearization error. 
Finally, $\mathbb{X}_\mathrm{f}^\varepsilon$ is chosen as a sublevel set of $V_\mathrm{f}^\varepsilon$ as in \eqref{eq:structure-Xf}, for a sufficiently small constant~$c > 0$.

Alternative methods for the design of terminal conditions satisfying Assumption~\ref{ass:terminal} for Koopman-based surrogate models is the use of quadratic functions in the lifted space~\cite{WortStra24} or the use of approximated Lyapunov functions computed from Koopman principle eigenfunctions \cite{deka2024extensions,vaidya2025koopman}.

\section{Kernel Extended Dynamic Mode Decomposition}\label{sec:implementation}

In this section, we recap 
kernel extended dynamic mode decomposition (kEDMD; \cite{williams2015data, klus2020kernel}) for autonomous nonlinear systems before presenting an extension for control-affine systems.
Finally, we modify the method to enforce the equilibrium condition and show that the resulting surrogate model respects Assumption \ref{ass:error_bounds}.
Let $\k :\R^n \times \R^n \rightarrow \R$ be a symmetric strictly positive definite kernel function, i.e., for all sets~$\mathcal{Y} = \{y_1, \dots, y_p\} \subset \R^n$ the corresponding kernel matrix~$K_\mathcal{Y} = (\k(y_i, y_j))_{i, j = 1}^p$ is positive definite. For $x \in \R^n$, we call the functions~$\Phi_x:\R^n \rightarrow \R$ with 
\begin{align*}
    \Phi_x(y) \coloneqq \k(x, y) 
\end{align*}
\textit{canonical features} of $\k$ and the completion of $\operatorname{span}\{\Phi_x : x \in \R^n\}$ 
yields a Hilbert space~$\mathbb{H}$ with inner product~$\langle \cdot, \cdot\rangle_\mathbb{H}$ induced by $\k$. Moreover, since for elements $\varphi \in \mathbb{H}$ the so-called \textit{reproducing property} is fulfilled, i.e., 
\begin{align*}
    \varphi(x) = \langle \varphi, \Phi_x\rangle_\mathbb{H}, \quad \forall \ x \in \R^n.
\end{align*}
$\mathbb{H}$ is referred to as \textit{reproducing kernel Hilbert space} (RKHS).\\

\noindent\textbf{Koopman operator and kernel EDMD}:
Let $\Omega \subset \R^n$ be open and bounded with Lipschitz boundary containing the origin in its interior. Moreover, consider a discrete-time nonlinear dynamical system described by  
\begin{align}\label{eq:sys:autonomous}
    x^+ = f(x)
\end{align}
with map~$f:\Omega \rightarrow \Omega$, where we additionally assume the domain $\Omega$ with $\mathbb{S} \subset \Omega$ to be forward invariant w.r.t.\ the dynamics~$f$, i.e., $f(x) \in \Omega$ for all $x \in \Omega$ to avoid technical difficulties, see, e.g., \cite{kohne2025error} and the references therein otherwise. Then, the associated linear Koopman operator~$\mathcal{K}:C_b(\Omega) \rightarrow C_b(\Omega)$ is defined for all observable functions~$\psi\in C_b(\Omega)$ by  
\begin{align*}
    (\mathcal{K}\psi)(\hat{x}) = \psi(f(\hat{x})) \quad \forall \ \hat{x} \in \Omega.
\end{align*}
In the following, we also assume that the kernel function~$\k$ is chosen such that the corresponding RKHS~$\mathbb{H}$ is invariant w.r.t. the Koopman operator, i.e., $\mathcal{K}\mathbb{H} \subseteq \mathbb{H}$. Kernel functions fulfilling said property are given by e.g. Wendland (see \eqref{eq:Wendland}) or Matérn kernels as shown in~\cite{kohne2025error}, that induce Hilbert spaces which coincide with fractional Sobolev spaces with equivalent norms \cite{FassYe11}. 
    In this paper, we work with kernel functions based on the Wendland radial basis functions (RBF)~$\Phi_{n, k}:\R^n \rightarrow \R$ with smoothness degree~$k \in \N$ from \cite{wendland2005approximate}. These functions induce a piecewise-polynomial and compactly-supported kernel function with 
    \begin{align}\label{eq:Wendland}
        \k(x, x') \coloneqq \Phi_{n, k}(x - y) = \phi_{n, k}(\|x - y\|) 
    \end{align}
    for $x, y, \in \R^n$ and $\phi_{n, k}$ defined as in \cite[Table 9.1]{Wend04}.

Let $\mathcal{X} = \{x_1, \dots, x_d\}\subset \Omega$ be a set of pairwise distinct data points, then set $V_\mathcal{X} = \operatorname{span}\{\Phi_{x_1}, \dots, \Phi_{x_d}\}$. 
As shown in \cite[Proposition 3.2]{kohne2025error}, a kEDMD approximant~$\widehat{K}$ of the Koopman operator~$\mathcal{K}$ is given by the orthogonal projection~$P_{V_\mathcal{X}}$ of $\mathbb{H}$ onto $V_\mathcal{X}$, with matrix approximant 
\begin{align*}
    \widehat{K} = K_\mathcal{X}^{-1}K_{f(\mathcal{X})}K_\mathcal{X}^{-1} \in \R^{d \times d},
\end{align*}
where $K_{f(\mathcal{X})} = (\k(f(x_i), x_j))_{i, j = 1}^d \in \R^{d \times d}$. 
The data-driven surrogate dynamics of \eqref{eq:sys:autonomous} can then be formulated for an observable~$\psi \in \mathbb{H}$ by
\begin{align*}
    \psi(f(x)) \approx \psi^+(x) \coloneqq \psi_\mathcal{X}^\top \widehat{K}^\top \mathbf{k}_\mathcal{X}(x),  
\end{align*}
where $\psi_\mathcal{X} = (\psi(x_1), \dots, \psi(x_d))^\top$ and $\mathbf{k}_\mathcal{X} = (\Phi_{x_1}, \dots, \Phi_{x_d})^\top$. 

In~\cite[Theorem~5.2]{kohne2025error} and~\cite[Theorem~3.7]{BoldPhil25}, a bound on the full approximation error was derived. 
Therein, the approximation accuracy of kEDMD depends on the the \textit{fill distance} of $\mathcal{X}$ in $\Omega$, defined by
\begin{align*}
    h_\mathcal{X} := \sup_{x \in \Omega} \min_{x_i \in \mathcal{X}} \|x - x_i \|.
\end{align*}

\textbf{Extension to control-affine systems}
Now, we consider control-affine systems of the form 
\begin{align}\label{eq:sys:control_affine}
    x^+ = f(x, u) = g_0(x) + G(x)u
\end{align}
with control input~$u \in \mathbb{U} \subset\R^m$, where $\mathbb{U}$ is compact with the origin in its interior, and with locally Lipschitz-continuous maps~$g_0:\R^n \rightarrow \R^n$ and $G:\R^n \rightarrow \R^{n \times m}$ defined by $G(x) \coloneqq [g_1(x) \mid \cdots \mid g_m(x) ]$.
An extension of data-driven kernel EDMD to systems in the form \eqref{eq:sys:control_affine} was provided in \cite[Section 4]{BoldPhil25}, which produces a method for using the autonomous kernel EDMD to approximate the functions~$g_0, g_1, \dots, g_m$. 
To do so, it is necessary to know the value of the functions $g_0, g_1, \dots, g_m$ in some observation points $x_i$, $i \in [1:d]$, which might not be 
accessible. 
Hence, we rely on an approximation of such values, that can be obtained relying on data points collected in the neighborhood of each observation point, as specified in the following assumption.
\begin{assumption}[Data requirements]\label{a:data}
    Let pairwise distinct cluster points $\mathcal{X} = \{x_1, \dots, x_d\} \subset \R^n$ with $x_1 = 0$ and a cluster radius $\rx \geq 0$ be given such that, for each $i \in [1:d]$, we have data $\mathcal{D}_i = \{(x_{ij}, u_{ij}, x_{ij}^+) : j \in [1:d_i]\}$, $d_i \geq m + (1 - \delta_{i1})$, satisfying ${x}_{ij} \in \mathcal{B}_{r_\mathcal{X}}(x_i) \cap \Omega$, $x_{ij}^+ = f(x_{ij}, u_{ij})$, and the rank condition $\operatorname{rank}(V_i) = m + 1$ for the matrix
    \begin{align*}   
        V_i = \begin{bmatrix}
            1 & \dots & 1 \\
            u_{i1} & \dots & u_{id_i}
        \end{bmatrix} \in \mathbb{R}^{(m+1) \times d_i}.
    \end{align*}
\end{assumption}

Let Assumption~\ref{a:data} hold. Then, we can compute approximations~$\tilde{g}_k(x_i)$ of the points~$g_k(x_i)$, $k \in [0:m]$. Since samples at $x_i$ are not required, we call $x_i$ a \emph{virtual observation point} in the following. 
Then, for each $x_i \in \mathcal{X}$, the data triplets $(x_{ij}, u_{ij}, x^+_{ij})$, 
$j \in [1:d_i]$, are used to compute an approximation $H_i := [\tilde{g}_0(x_i) \mid \tilde{G}(x_i)] \in \mathbb{R}^{n \times (m+1)}$ of $[g_0(x_i) \mid G(x_i)]$  by solving the linear regression problem 
\begin{align}\label{eq:regression:step1}
    \operatorname{arg} \min_{H_i} \big\| [x^+_{i1} \mid \ldots \mid x^+_{id_i}] - H_i V_i \big\|^2_\mathcal{F}.
\end{align}
Given an arbitrary observable function $\psi\in \mathbb{H}$, the propagation of the observable can be conducted as follows
\begin{align*}
    \psi(f(x, u)) \approx \psi^{+,\varepsilon}(x, u) \coloneqq \psi_\mathcal{X}^\top \Big(\widehat{K}_0 + \sum_{k = 1}^m \widehat{K}_k u_k\Big)^\top \mathbf{k}_\mathcal{X}(x),
\end{align*}
with $\widehat{K}_k = K_\mathcal{X}^{-1}K_{\tilde{g}_k(\mathcal{X})}K_\mathcal{X}^{-1}$, $k \in [0:m]$, where $K_{\tilde{g}_k(\mathcal{X})} = (\k(\tilde{g}_k(x_i), x_j))_{i, j = 1}^{d_i}$. We then obtain the full state-space surrogate model 
\begin{align*}
    x^+ = f^\varepsilon(x, u) = \begin{pmatrix}
        \psi_1^{+, \varepsilon}(x, u) \\ \vdots \\ \psi_{n}^{+, \varepsilon}(x, u)
    \end{pmatrix}
\end{align*}
by using the observable functions~$\psi_\ell(x) = e_\ell^\top x$, the $\ell-$th coordinate function for $\ell \in [1:n]$. 

An error bound for this surrogate model was derived in \cite[Theorem 3]{schimperna2025data}, which can be made arbitrarily small by suitably decreasing the fill distance $h_\mathcal{X}$ and the cluster radius $r_\mathcal{X}$ of the dataset. 
To obtain also a proportional error bound, a modification of the algorithm, as proposed in~\cite[Section 4]{schimperna2025data}, is implemented by modifying the regression problem~\eqref{eq:regression:step1} at $x_1 = 0$ by enforcing that $\tilde{g}_0(x_1) = 0$. The following theorem, see \cite[Proposition 2]{schimperna2025data}, gives said proportional bound on the full approximation error and therefore shows that the surrogate model $f^\varepsilon$ fulfills Property 1 of Assumption~\ref{ass:error_bounds}.
\begin{theorem}[Approximation error; Proposition 2 in \cite{schimperna2025data}]\label{thm:control_main}
    Let Assumption~\ref{a:data} on the data hold and $\mathbb{H}$ be the RKHS induced by the Wendland kernel~$\k$ with smoothness degree~$k \ge 1$. 
    In addition, let $f \in C_b^{\ceil{\sigma_{n,k}}}(\Omega,\R^{n})$ with $\sigma_{n,k} := \frac{n+1}{2} + k$ be given. Then, there exist constants $C_1, C_2, C_3, C_4, C_5, h_0, r_0>0$ such that, if the fill distance~$h_\mathcal{X}$ and the cluster radius~$r_\mathcal{X}$ satisfy $h_\mathcal{X} \leq h_0$ and $r_\mathcal{X} \leq r_0$, we have the proportional error bound
    \begin{equation}\label{eq:prop:bound}
        \|f(x,u) - f^\varepsilon(x,u)\| \leq c_x^\varepsilon \|x\| + c_u^\varepsilon \|u\|
    \end{equation}
    for all $(x, u) \in \Omega \times \mathbb{U}$ with
    \begin{align*}
        c_u^\varepsilon & := 
        C_1 \hspace*{3.1cm} h_\mathcal{X}^{k+1/2} + C_2\|K_\mathcal{X}^{-1}\|\hspace*{0.28cm}\rx \\
        c_x^\varepsilon &:= \Big[ C_3 + m C_4\max_{u \in \mathbb{U}} \|u\|_\infty \Big] h_\mathcal{X}^{k-1/2} + \tilde{C}(\mathcal{D})\hat{C}(\mathcal{X})\, r_\mathcal{X}
    \end{align*}
    with $\mathcal{D} = \bigcup_{i \in [1:d]} \mathcal{D}_i$ and functions
    \begin{align*}
        \tilde{C}(\mathcal{D}) &\coloneqq C_5  
        (\max_{u \in \mathbb{U}}\|u\|_1 + 1)
        \max_{k \in [1:d]}\sqrt{2d_k}\|V_k^\dagger\| \\
        \hat{C}(\mathcal{X}) &\coloneqq n \Phi_{n, k}(0) \cdot \|K_\mathcal{X}^{-1}\| \left(
        \max_{v \in \mathbb{R}^d: |v_i| = 1 }v^\top K_\mathcal{X}^{-1}v \right).
    \end{align*}
\end{theorem}
Sufficient conditions to ensure that the minimal singular value of $V_i$, $i \in [1:d]$, is uniformly bounded from below even for varying~$\varepsilon$ are given in \cite[Theorem 3.6]{SchmBold25} such that $\tilde{C}(\mathcal{D})$ can be uniformly bounded, whereas $\hat{C}(\mathcal{X})$ has to be compensated by suitably adapting the cluster radius~$r_\mathcal{X}$.

This approach, which guarantees the existence of proportional error bounds related to $h_\mathcal{X}$ and $r_\mathcal{X}$, 
is denoted in the following by \textit{physics-informed kEDMD} (\textbf{PI-kEDMD}). Instead, we use \textbf{kEDMD} to denote models obtained with the standard identification algorithm described in~\cite{BoldPhil25}. 
Finally, we point out that Lipschitz continuity of kEDMD models is proven in~\cite[Section 4]{schimperna2025data} such that we can constructively verify the required Assumption~\ref{ass:error_bounds}.

\section{Numerical example} \label{sec:examples}

In this section, the proposed MPC framework is illustrated for the Euler discretization of the controlled van der Pol-oscillator, given by the control affine system
\begin{align}\label{eq:sys:vdp}
    x^+ = x + \Delta t \begin{pmatrix}
        x_2 \\ \nu (1-x_1^2)x_2 - x_1 + u
    \end{pmatrix}
\end{align}
for a time step~$\Delta t = 0.05$ and parameter~$\nu = 0.1$. We consider the domain where we sample to be $\Omega = [-2, 2]^2$ and the input constraint set~$\mathbb{U} = [-2, 2]$. 
As kernel functions, we use the Wendland function given by
    \begin{align*}
        \phi_{n, 1}(r) &= \begin{cases}
              \frac{1}{20}(1-r)^4(4r + 1) & \text{for }\ 0 \leq r \leq 1 \\
              0 & \text{else } 
            \end{cases}. 
    \end{align*}
We consider virtual observation points arranged in a grid of Padua points, as described in~\cite{caliari2005bivariate}, to which we add a point $x_1 = 0$. In particular, we consider $d \in \{352, 1327\}$ virtual observation points.
For each virtual observation point, we sample $d_i = 25$ data points~$(x_{ij}, u_{ij}, x^+_{ij}) \in \mathcal{B}_{r_\mathcal{X}}(x_i)$, $i \in [1:d]$, considering a cluster radius~$r_\mathcal{X} = \frac{\sqrt{2}}{d}$. 
For PI-kEDMD, the regression problem~\eqref{eq:regression:step1} is modified for $x_1$ in order to obtain proportional error bounds.
For the MPC implementation we considered a domain $\mathbb{S} = [-1.9, 1.9]^2$,
on which we estimated uniform error bounds of $\eta^\varepsilon = 0.05$ for $d = 352$ and of $\eta^\varepsilon = 0.005$ for $d = 1327$. The use of a reduced domain compared to $\Omega$ is motivated by the larger values of the error in $\Omega \setminus \mathbb{S}$.
These values have been used for the constraint tightening from Definition~\ref{def:admissibility}, together with the estimated Lipschitz constants of~$\bar L = 2.27$ for $d = 352$ and $\bar L = 1.75$ for $d = 1327$. 
For these values, the maximum prediction horizons for which the tightened sets are nonempty are $N = 4$ for $d = 352$ and $N = 9$ for $d = 1327$. Therefore, horizons of $N = 4$, and  matrices $Q = I \in \R^{2 \times 2}$ and $R = 10^{-4} \in \R$ are chosen for the simulations. From the initial state~$\hat{x} = (0.5, 0.5)^\top$ we then want to steer the system to the origin. 
The simulations are implemented in Python, using Casadi to implement the solution of the OCP.

Figure \ref{fig:vdp} compares the simulation results for different numbers of virtual observation points and the two kEDMD methods. It is visible that in absence of proportional error bounds the error stagnates at a positive constant value, i.e. only practical asymptotic stability is obtained. Instead, the MPC based on PI-kEDMD models achieves a much higher accuracy.

\begin{figure}
    \centering
    \includegraphics[width=0.5\textwidth]{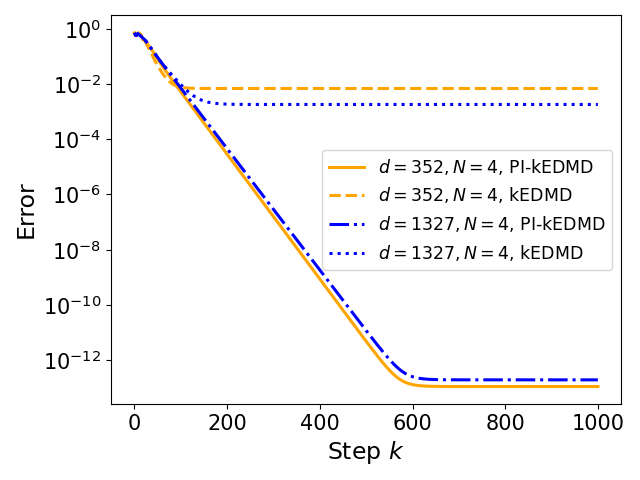}
    \caption{Van der Pol \eqref{eq:sys:vdp}: Error $\|x(k)\|$ of the kEDMD-based MPC with different numbers of clusters.}
    \label{fig:vdp}
\end{figure}

\section{Conclusion} \label{sec:conclusion}
In this paper, we have proved that MPC using a surrogate model in the 
optimization step can asymptotically stabilize the system if proportional and uniform error bounds holds. For a large class of nonlinear systems, these conditions can be satisfied using kEDMD surrogate models.

The main drawback of the proposed approach is the conservativeness of the constraint tightening based on the Lipschitz constant, in which the tightening grows unbounded with increasing prediction horizons. This issue could be mitigated using more sophisticated constraint tightening techniques, see, e.g., \cite{kohler2018novel}, which will be considered for future works.

\bibliographystyle{ieeetr}
\bibliography{references_ECC}

\end{document}